\documentclass[preprint, 12pt]{elsarticle}
\usepackage{amssymb}
\usepackage{graphicx}
\usepackage{float}
\usepackage{algorithm}
\usepackage{algorithmicx}
\usepackage{algpseudocode}
\usepackage{listings}
\usepackage{hyperref}
\usepackage{fullpage}
\journal{Journal of Discrete Algorithms}
\newtheorem{theorem}{Theorem}[section]
\newtheorem{proof}{Proof}[theorem]
\newtheorem{lemma}{Lemma}[theorem]
\newtheorem{proposition}{Proposition}[theorem]
\newtheorem{justification}{Justification}[theorem]
\begin{document}
\begin{frontmatter}
\title{Uniqueness Trees: A Possible Polynomial Approach to the Graph Isomorphism Problem}
\author{Jonathan Gorard}
\address{Department of Mathematics, King's College London, Strand, London, WC2R 2LS}
\begin{abstract}
This paper presents the novel `uniqueness tree' algorithm, as one possible method for determining whether two finite, undirected graphs are isomorphic. We prove that the algorithm has polynomial time complexity in the worst case, and that it will always detect the presence of an isomorphism whenever one exists. We also propose that the algorithm will equivalently discern the lack of an isomorphism whenever one does not exist, and some initial justifications are given for this proposition, although it cannot yet be rigorously proven. Finally, we present experimental evidence for both the effectiveness and efficiency of the uniqueness tree method, using data gathered from a practical implementation of the algorithm. Some consequences and directions for further research are discussed.
\end{abstract}
\end{frontmatter}

\section{Introduction}

\subsection{Background\\}%

The graph isomorphism problem is the decision problem of determining whether two finite graphs, ${G=(V,E)}$ and ${H=(U,F)}$ are isomorphic, denoted ${G \cong H}$. The graphs are isomorphic if and only if there exists a bijection $f$ between the two sets of vertices $V$ and $U$ such that, for every pair of vertices ${(u,v)}$ in $V$, the edge ${f(u)f(v)}$ exists in $F$ if and only if the corresponding edge ${uv}$ exists in E.\cite{prolubnikov} Formally:

\begin{equation}
G \cong H \iff \exists f: V \rightarrow U \mid \forall (u, v) \in V, f(u)f(v) \in F \iff uv \in E
\end{equation}

The graph isomorphism problem is of particular interest in the field of computational complexity theory, since it is one of only a few problems whose complexity class is not solidly classified: it is not known to be solvable in polynomial time, yet neither has it been shown to be NP-complete. Thus, it is often placed in the theoretical complexity class of `${NP}$-intermediate'\cite{dawar} (which, by Ladner's theorem, exists if and only if ${P \neq NP}$\cite{bodirsky}). In addition, efficient algorithms for detecting graph isomorphism are of great practical importance across a variety of fields, including network analysis, organic chemistry, condensed matter physics, chemical engineering, electronic engineering, computational biology, and others.\cite{akutsu}\cite{fan}\cite{whitham}

There exist many algorithms which run in polynomial (or even sub-polynomial) time in all practical cases, but which degenerate to exponential time in the worst-case\cite{foggia}, making them unsatisfactory from a complexity-theoretic point of view. We provide a brief outline of one such algorithm (McKay's NAUTY), for the purpose of demonstrating how its approach, and the approach adopted by many similar state-of-the-art algorithms, differs from the uniqueness tree method proposed in this paper.

We then provide both a formal and an informal statement of the uniqueness tree algorithm itself, along with an illustrative example of its application to a pair of non-isomorphic graphs. Next, we give a formal proof that the algorithm runs in septic polynomial (${O(n^7)}$) time in the worst case, and that the algorithm will always correctly detect the presence of an isomorphism between two graphs, whenever one exists. We propose that the converse statement is also true (i.e. that the algorithm will correctly discern the lack of an isomorphism, whenever one does not exist), and give a brief sketch of a possible proof method, though this statement remains a conjecture. Finally, we supply experimental evidence of both the algorithm's effectiveness in determining isomorphism/non-isomorphism between random graphs, and its polynomial efficiency.

For the purposes of this paper, we will consider only simple graphs (i.e. unweighted, undirected graphs containing no loops or multiple edges). However, it is possible to generalise these methods to directed graphs, as well as to graphs in which multiple edges and loops are permitted, as will be shown in a future work.

\subsection{NAUTY\\}

Most practical graph isomorphism algorithms work by reducing graphs to a so-called `canonical form' - an object whose structure is independent of the particular ordering of the vertices, but dependent upon all other properties of the graph.\cite{mckay1} Thus, if the canonical forms for two graphs are equivalent, then the graphs must be isomorphic; conversely, if the canonical forms differ, then the graphs must be non-isomorphic. NAUTY, as with most similar algorithms, operates by producing a so-called `search tree' as its canonical graph form, which is a rooted tree in which each vertex corresponds to a distinct `partition' (colouring) of the graph's vertices. Loosely, the process of partition refinement is as follows:\cite{mckay2}

\begin{figure}[H]
\fbox{
\begin{minipage}[B]{0.45\linewidth}
\centering
\includegraphics[width=100pt]{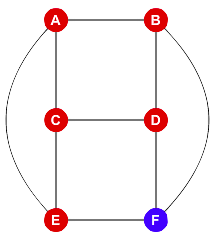}
\caption{A partition ${\pi}$.}
\end{minipage}
}
\fbox{
\begin{minipage}[B]{0.45\linewidth}
\centering
\includegraphics[width=100pt]{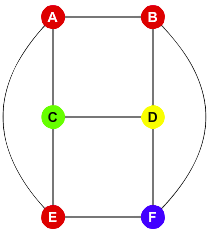}
\caption{A refined partition ${\pi_1}$.}
\end{minipage}
}
\end{figure}

\begin{enumerate}
\item The root of the tree is an initial (uniform) colouring of the graph.\\
\item If two vertices share the same colour in a particular partition, but have neighbourhoods with different colourings, then `refine' the partition by assigning each vertex a new colour. This refined partition is a child vertex of the original partition.\\
\item If a particular partition cannot be refined further, then that vertex of the tree becomes a leaf, with no children.
\end{enumerate}

However, the uniqueness tree algorithm does not produce a single search tree to represent a graph. Rather, it produces a set of `uniqueness trees': one for each vertex of the graph. In turn, each vertex of a uniqueness tree represents a single vertex of the graph, as opposed to the richer structure of an entire vertex partition. This paper conjectures that the set of uniqueness trees is a canonical graph form, and provides some initial justification for that assertion.

\section{The Uniqueness Tree Algorithm}

\subsection{Brief Outline\\}

For two finite, simple graphs ${G = (V, E)}$ and ${H = (U, F)}$, the uniqueness tree algorithm associates a rooted tree ${T(v)}$ with every vertex ${v \in V}$, and ${T(u)}$ with every vertex ${u \in U}$. This paper proposes that, if every tree associated with a vertex of $G$ is uniquely isomorphic to a tree associated with a vertex of $H$, then $G$ and $H$ are isomorphic:

\begin{equation}
G \cong H \iff \forall v \in V, \exists! u \in U \mid T(v) \cong T(u)
\end{equation}

Checking this criterion is an efficient process, since an isomorphism between two rooted trees may be computed in linear (${O(n)}$) time, where $n$ is the number of vertices in each tree.\cite{buss}

The process for generating the uniqueness tree for a vertex $v$ is as follows:
\begin{enumerate}
\item The root of the tree is the vertex $v$ itself.\\
\item Every vertex in the current level of the tree which is not `unique' (i.e. every vertex which appears in the current level more than once) becomes a leaf, producing no children.\\
\item Every unique vertex in the current level of the tree produces 1 child for every adjacent vertex in the graph.\\
\item This process continues until either the tree self-terminates (i.e. there are no more unique vertices on the current level), or the height of the tree reaches $n$ (the size of the graph).
\end{enumerate}

\subsection{An Example Case\\}

As an illustrative example, we shall test for isomorphism between the following pair of graphs:

\begin{figure}[H]
\fbox{
\begin{minipage}[B]{0.45\linewidth}
\centering
\includegraphics[height=160pt]{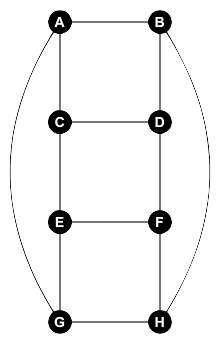}
\caption{A graph, ${G = (V, E)}$.}
\end{minipage}
}
\fbox{
\begin{minipage}[B]{0.45\linewidth}
\centering
\includegraphics[height=160pt]{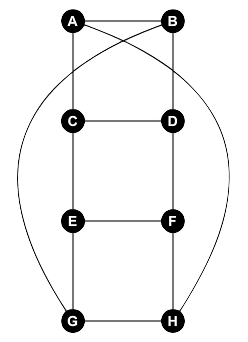}
\caption{Another graph, ${H = (U, F)}$.}
\end{minipage}
}
\end{figure}

Clearly, $G$ and $H$ are not isomorphic ($G$ is planar and $H$ is not). However, they are equivalent in every other respect, since $H$ is simply an embedding of $G$ from a plane onto a M\"obius strip. In the interests of brevity, we shall not apply the entire algorithm, since to do so would require generating 16 uniqueness trees, and then making up to 36 comparisons between them. Rather, we shall simply show that vertex $A$ in $G$ cannot be equivalent to vertex $A$ in $H$, and the non-isomorphism of $G$ and $H$ follows trivially.

For graph $G$, the first level of $A$'s uniqueness tree contains $A$'s immediate adjacencies: $B$, $C$ and $G$ (all of which are unique).

\begin{figure}[H]
\centering
\fbox{
\begin{minipage}[B]{0.9\textwidth}
\centering
\includegraphics[height=80pt]{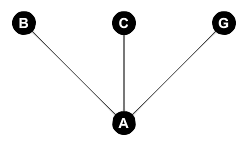}
\caption{The first level of $A$'s uniqueness tree for graph $G$.}
\end{minipage}
}
\end{figure}

Similarly, $B$'s children on the second level of the tree are $A$, $D$ and $H$ ($B$'s immediate adjacencies), $C$'s children are $A$, $D$ and $E$, and $G$'s children are $A$, $E$ and $H$. Since $A$, $D$, $E$ and $H$ all appear multiple times on the second level, all of the vertices becomes leaves and the tree terminates.

\begin{figure}[H]
\centering
\fbox{
\begin{minipage}[B]{0.9\textwidth}
\centering
\includegraphics[height=100pt]{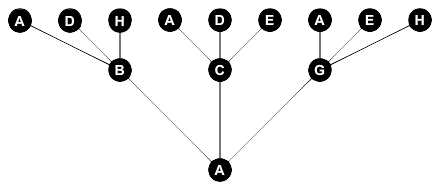}
\caption{The complete uniqueness tree for vertex $A$ in graph $G$.}
\end{minipage}
}
\end{figure}

On the other hand, the first level of $A$'s uniqueness tree for graph $H$ contains the adjacent vertices $B$, $C$ and $H$, which are, again, unique.

\begin{figure}[H]
\centering
\fbox{
\begin{minipage}[B]{0.9\textwidth}
\centering
\includegraphics[height=80pt]{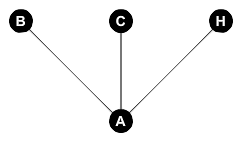}
\caption{The first level of $A$'s uniqueness tree for graph $H$.}
\end{minipage}
}
\end{figure}

$B$'s children on the second level of the tree are then $A$, $D$ and $G$, $C$'s children are $A$, $D$ and $E$, and $H$'s children are $A$, $F$ and $G$.  Vertices $A$, $D$ and $G$ become leaves, since they are not unique at this level.

\begin{figure}[H]
\centering
\fbox{
\begin{minipage}[B]{0.9\textwidth}
\centering
\includegraphics[height=100pt]{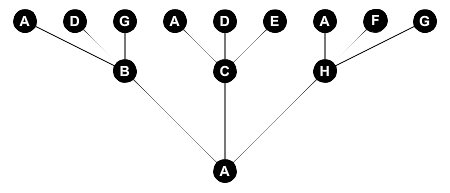}
\caption{The second level of $A$'s uniqueness tree for graph $H$.}
\end{minipage}
}
\end{figure}

Since vertices $E$ and $F$ are both unique on the second level, their adjacencies ($C$, $F$ and $G$ for vertex $E$, and $D$, $E$ and $H$ for vertex $F$) are carried to the third level of the tree. Since $C$, $D$, $E$, $F$, $G$ and $H$ are all unique on the third level, their adjacencies will, in turn, be carried up to the fourth level, and so on.

\begin{figure}[H]
\centering
\fbox{
\begin{minipage}[B]{0.9\textwidth}
\centering
\includegraphics[height=150pt]{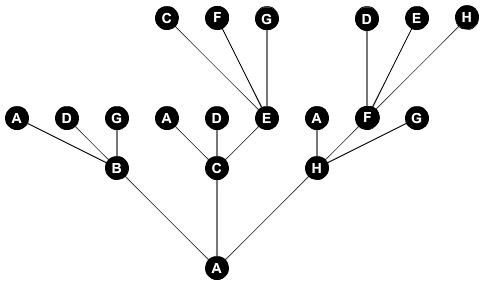}
\caption{The third level of $A$'s uniqueness tree for graph $H$.}
\end{minipage}
}
\end{figure}

Already, we can see that the two uniqueness trees cannot possibly be isomorphic (since the first tree has only two levels, whilst the second has at least four), and so vertex $A$ in graph $G$ cannot be equivalent to vertex $A$ in graph $H$.

\subsection{Formal Statement\\}

The uniqueness tree algorithm may be divided into two distinct stages: the tree generation stage, and the tree comparison stage.

\begin{algorithm}[H]
\caption {Uniqueness tree generation}
\begin{algorithmic}[1]
\For {each graph ${G = (V,E)}$}
\For {each vertex ${v \in V}$}
\State {create a new tree ${T(v)}$, with $v$ as its root}
\While {${\exists}$ $u$, a vertex which appears only once in the current level of ${T(v)}$, and ${height(T(v)) < n}$}
\For {each ${w \in neighbourhood(u)}$}
\State {add $w$ to the next level of the uniqueness tree.}
\EndFor
\EndWhile
\EndFor
\EndFor
\end{algorithmic}
\end{algorithm}

Then, for two graphs ${G = (V, E)}$ and ${H = (U, F)}$:

\begin{algorithm}[H]
\caption {Uniqueness tree comparison}
\begin{algorithmic}[1]
\For {each unmapped vertex ${v \in V}$}
\For {each unmapped vertex ${u \in U}$}
\If {${height(T(v)) \neq height(T(u))}$}
\State {$v$ and $u$ are not equivalent}
\EndIf
\For {each level of ${T(v)}$}
\If {${width(currentlevel(T(v))) \neq width(currentlevel(T(u)))}$}
\State {$v$ and $u$ are not equivalent}
\EndIf
\For {$i \gets 1, (n - 1)$}
\If {vertices with $i$ children in ${currentlevel(T(v)) \neq}$ vertices with $i$ children in ${currentlevel(T(u))}$}
\State {$v$ and $u$ are not equivalent}
\EndIf
\EndFor
\EndFor
\If {$v$ and $u$ are equivalent}
\State {map $v$ onto $u$}
\EndIf
\EndFor
\EndFor
\If {all vertices ${v \in V}$ and ${u \in U}$ have been mapped}
\State {${G \cong H}$}
\Else
\State {${G \ncong H}$}
\EndIf
\end{algorithmic}
\end{algorithm}

\section{Rigorous Results}

\subsection{Proof of Polynomial Time Complexity\\}

\begin{theorem}
For two finite, simple graphs ${G=(V,E)}$ and ${H=(U,F)}$, each of size $n$, the uniqueness tree algorithm runs in septic polynomial time (${O(n^7)}$) in the worst case.
\end{theorem}

Since the uniqueness tree algorithm can be divided into two sequential stages (tree generation and tree comparison), we shall analyse the time complexity of each stage separately, and then add the two complexities together.

\begin{lemma}
The tree generation algorithm runs in sextic polynomial time (${O(n^6)}$) in the worst case.
\end{lemma}

\begin{proof}
\hfill
\begin{enumerate}
\item The total number of uniqueness trees which must be generated is ${2n = O(n)}$.\\
\item The maximum width of a single uniqueness tree is ${O(n^2)}$ (since the maximum number of unique vertices which can appear in a single level is $n$, and each vertex can have a maximum of ${(n - 1)}$ adjacencies, giving a maximum of ${(n^2 - n)}$ children on the next level, all of which would be leaves).\\
\item ${\therefore}$ The maximum number of vertices in each uniqueness tree is ${O(n^3)}$ (since the maximum width is ${O(n^2)}$, and the maximum height is ${O(n)}$).\\
\item The maximum number of operations required to generate each vertex in the tree is ${O(n^2)}$ (i.e. a maximum of ${(n - 1)}$ comparisons with other vertices in the graph, plus a maximum of ${(n^2 - n)}$ comparisons with other vertices in the current level of the tree, in order to test uniqueness).\\
\item ${\therefore}$ The worst case time complexity of the tree generation algorithm is

${O(n) * O(n^3) * O(n^2) = O(n^6)}$. ${\blacksquare}$
\end{enumerate}
\end{proof}

\begin{lemma}
The tree comparison algorithm runs in septic polynomial time ${O(n^7)}$ in the worst case.
\end{lemma}

\begin{proof}
\hfill
\begin{enumerate}
\item The first tree of graph $G$ must be compared with a maximum of $n$ trees from graph $H$, the second with ${(n-1)}$ trees, and so on.

${\therefore}$ The maximum number of tree comparisons which must be made is ${\displaystyle\sum\limits_{i=1}^n i = O(n^2)}$.\\
\item Determining an isomorphism between the two rooted trees requires comparing the heights of both trees, comparing the total number of vertices in each level of each tree, and comparing the total number of children possessed by each vertex in each level of each tree.\\
\item Comparing the heights of two rooted trees requires ${O(1)}$ operation.\\
\item Comparing the total number of vertices in each level of two rooted trees requires ${O(n)}$ operations in the worst case (since each tree may hold a maximum of $n$ levels).\\
\item Each level may contain a maximum of ${(n^2 - n)}$ vertices, each of which must be compared with a maximum of ${(n^2 - n)}$ vertices in the corresponding level of the other tree, for each of a possible $n$ levels of the tree.

${\therefore}$ Comparing the total number of children possessed by each vertex in each level of each tree requires

${O(n^2) * O(n^2) * O(n) = O(n^5)}$ operations in the worst-case.\\
\item ${\therefore}$ The worst case time complexity for detecting an isomorphism between the two rooted trees is

${O(n^5) + O(n) + O(1) = O(n^5)}$.\\
\item ${\therefore}$ The worst case time complexity of the tree comparison algorithm is

${O(n^2) * O(n^5) = O(n^7)}$. ${\blacksquare}$
\end{enumerate}
\end{proof}

Since ${O(n^7) + O(n^6) = O(n^7)}$, the desired theorem follows directly from these two lemmas.

\subsection{Proof of Effectiveness (Positive Case)\\}

\begin{theorem}
If two graphs ${G = (V, E)}$ and ${H = (U, F)}$ are isomorphic, then the uniqueness tree algorithm will correctly determine that ${G \cong H}$.
\end{theorem}

\begin{proof}

\hfill

The key point is that the uniqueness tree algorithm is based entirely around vertex adjacencies (which are graph-invariant), and does not depend at all upon vertex ordering. Suppose that ${uv \in E}$ and ${f(u)f(v) \in F}$.

\begin{figure}[H]
\fbox{
\begin{minipage}[B]{0.45\textwidth}
\centering
\includegraphics[height=70pt]{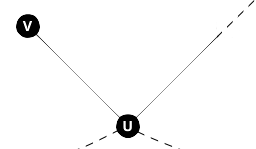}
\caption{A section of a uniqueness tree for graph $G$.}
\end{minipage}
}
\fbox{
\begin{minipage}[B]{0.45\textwidth}
\centering
\includegraphics[height=70pt]{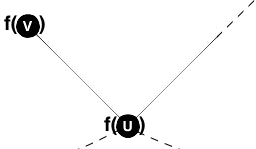}
\caption{A section of the equivalent uniqueness tree for graph $H$.}
\end{minipage}
}
\end{figure}

\begin{enumerate}
\item If a vertex $u$ is not unique to a particular level of a particular uniqueness tree for graph $G$, then the vertex ${f(u)}$ will be also non-unique to the equivalent level of the equivalent vertex tree for graph $H$.\\
\item Conversely, if vertex $u$ is unique to that level, then vertex $v$ in the next level of $G$'s tree will have the same uniqueness property as vertex ${f(v)}$ in $H$'s tree, and so on.\\
\item Thus, it follows that if ${f(u)f(v) \in F \iff uv \in E}$, then graphs $G$ and $H$ will produce an equivalent set of uniqueness trees.\\
\item Since an isomorphism is defined as a bijection $f$ between the sets $V$ and $U$ which satisfies the above property, the desired theorem follows naturally. ${\blacksquare}$
\end{enumerate}
\end{proof}

\subsection{Proposition of Effectiveness (Negative Case)\\}

\begin{proposition}
If two graphs ${G = (V, E)}$ and ${H = (U, F)}$ are not isomorphic, then the uniqueness tree algorithm will correctly determine that ${G \ncong H}$.
\end{proposition}

\begin{justification}

\hfill

One possible approach to proving this proposition may be an analogy to the positive case, as in the following informal sketch:

\begin{enumerate}
\item If ${G \ncong H}$, then (without loss of generality) there must exist at least one ${f(u)f(v) \in F}$, such that ${uv \notin E}$ (by negation of the definition of isomorphism).\\
\item If vertex $u$ (in a particular level of a particular uniqueness tree for graph $G$) has a different uniqueness property to vertex ${f(u)}$ (in the equivalent level of the equivalent uniqueness tree for graph $H$), then the uniqueness trees produced will be non-isomorphic, since one vertex will become a leaf, whilst the other will not.\\
\item Conversely, if vertices $u$ and ${f(u)}$ are both unique, then there will be an instance of the vertex ${f(v)}$ in the next level of $H$'s tree, without a corresponding instance of the vertex $v$ in $G$'s. Thus, the uniqueness trees produced will be non-isomorphic, since the vertices in the next level will have different uniqueness properties as a result.\\
\item Additionally, if vertices $u$ and ${f(u)}$ are both non-unique, then the non-uniqueness of vertex ${f(u)}$ may be the result of having ${f(v)}$ as its parent vertex in $H$'s tree, whilst $v$ cannot be a parent vertex of $u$ in $G$'s tree. If this is the case, then the uniqueness trees produced will, again, be non-isomorphic.
\end{enumerate}
\end{justification}

Clearly, further work will be required to either formalise a proof of this proposition, or to demonstrate why it is incorrect.

\section{Practical Implementation}

For the purposes of experimentally verifying both the effectiveness and efficiency of the uniqueness tree method, we present a practical implementation of the algorithm in Java; the complete source code for this algorithm, along with all of the tests described in this section, may be found in the Appendix. Independent verification of these experimental findings is strongly encouraged.

\subsection{Experimental Verification of Effectiveness (Positive Case)\\}

The algorithm was tested on 10,000 pairs of random graphs, ranging in size from 1 to 100 (100 pairs of each). The graph pairs were known to be isomorphic, since the second graph was generated by randomly permuting the vertices of the first. The presence of an isomorphism was correctly detected in all 10,000 cases.

\subsection{Experimental Verification of Effectiveness (Negative Case)\\}

When testing proposed algorithms for graph isomorphism, generating pairs of graphs which are known to be non-isomorphic is a fundamentally difficult problem, since doing so pre-supposes an effective criterion for detecting isomorphisms in the first place. It is trivial to generate pairs of graphs with differing graph invariant properties, but then one is simply testing the capability of the algorithm to discern those particular graph invariants, rather than its ability to detect non-isomorphism in general.

For the purpose of the present experimental test, we will attempt to generate a pair of similar but non-isomorphic graphs, by creating a second graph that is a vertex-permutation of the first (as above), but with one edge randomly replaced. Clearly, there is a small probability of inadvertently creating an isomorphic pair, particularly when the graphs involved are small. Indeed, when tested on the first 900 pairs of graphs with ${n < 10}$, some apparent false-positives were produced. However, when these possible exceptions were tested (either by hand, or by cross-checking with NAUTY), it was determined that all of the problematic graphs had been, by chance, isomorphic. For the remaining 9,100 pairs of graphs with ${n \geq 10}$, there were no false positives.

\subsection{Experimental Verification of Efficiency\\}

The total computation time required for the algorithm to test for isomorphism between 10,000 pairs of random isomorphic test graphs (100 pairs of each size $n$, with $n$ ranging from 1 to 100) was recorded, and then plotted as a function of $n$. The equivalent process was repeated for 10,000 pairs of random non-isomorphic test graphs, also:

\begin{figure}[H]
\fbox{
\begin{minipage}[B]{0.45\textwidth}
\centering
\includegraphics[height=180pt]{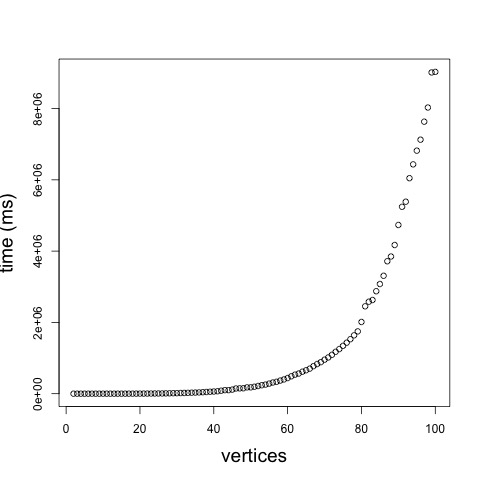}
\caption{The computation time required to test for isomorphism between 100 pairs of known \textit{isomorphic} graphs, as a function of the number of vertices.}
\end{minipage}
}
\fbox{
\begin{minipage}[B]{0.45\textwidth}
\centering
\includegraphics[height=180pt]{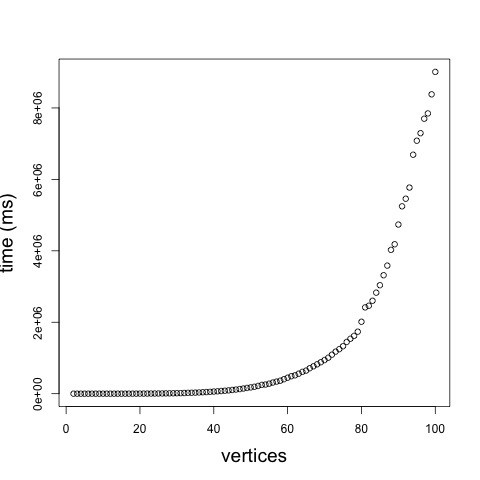}
\caption{The computation time required to test for isomorphism between 100 pairs of known \textit{non-isomorphic} graphs, as a function of the number of vertices.}
\end{minipage}
}
\end{figure}

For large $n$, taking the natural logarithm of both axes produces a straight line, demonstrating the asymptotically polynomial efficiency of the algorithm (since if ${y = x^k}$, then ${ln(y) = k ln(x)}$, with the gradient of the line thus representing the degree of the polynomial):

\begin{figure}[H]
\fbox{
\begin{minipage}[B]{0.45\textwidth}
\centering
\includegraphics[height=180pt]{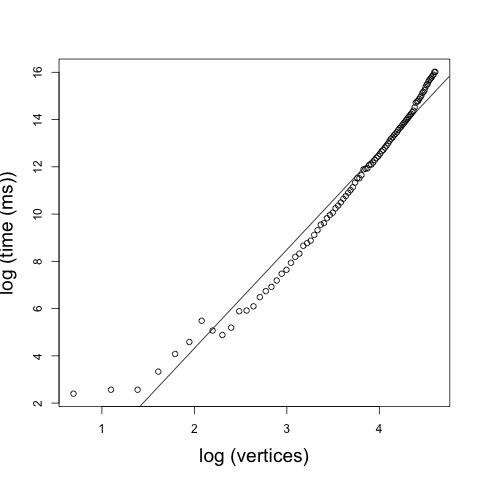}
\caption{The natural logarithm of the computation time required to test for isomorphism between 100 pairs of \textit{isomorphic} graphs, as a function of the natural logarithm of the number of vertices. The plot fits a straight line of gradient 4.174996, with ${R^2 = 0.968827}$.}
\end{minipage}
}
\fbox{
\begin{minipage}[B]{0.45\textwidth}
\centering
\includegraphics[height=180pt]{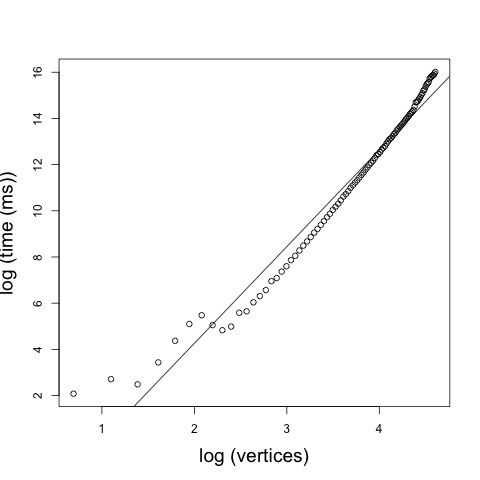}
\caption{The natural logarithm of the computation time required to test for isomorphism between 100 pairs of \textit{non-isomorphic} graphs, as a function of the natural logarithm of the number of vertices. The plot fits a straight line of gradient 4.187919, with ${R^2 = 0.966003}$.}
\end{minipage}
}
\end{figure}

The seemingly anomalous points observed for low values of $n$ are due to intrinsic computational overheads, which naturally smooth out for larger graphs. Thus, removing the results for the first 2,000 graph pairs from each plot (i.e. the values ${1 \leq n \leq 20}$) produces a much better linear fit in both cases:

\begin{figure}[H]
\fbox{
\begin{minipage}[B]{0.45\textwidth}
\centering
\includegraphics[height=180pt]{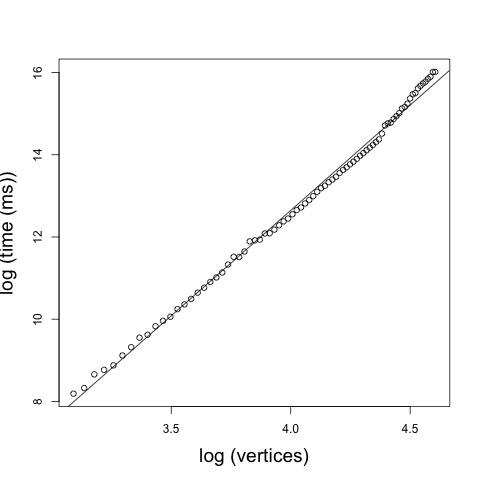}
\caption{The same as the plot above, but with the results for the smallest 2,000 \textit{isomorphic} graph pairs removed. The plot now fits a straight line of gradient 5.124289, with ${R^2 = 0.995931}$.}
\end{minipage}
}
\fbox{
\begin{minipage}[B]{0.45\textwidth}
\centering
\includegraphics[height=180pt]{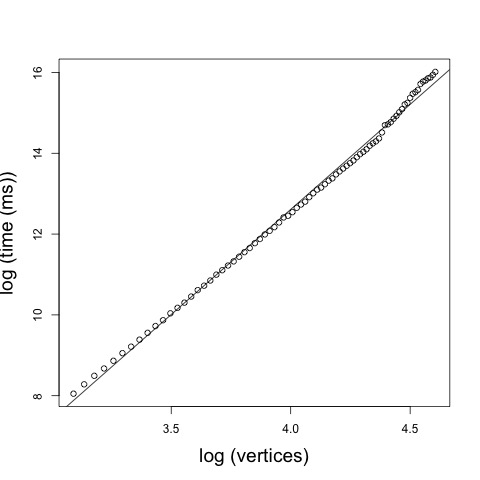}
\caption{The same as the plot above, but with the results for the smallest 2,000 \textit{non-isomorphic} graph pairs removed. The plot now fits a straight line of gradient 5.202082, with ${R^2 = 0.996664}$.}
\end{minipage}
}
\end{figure}

These plots suggest that the uniqueness tree algorithm runs in approximately quintic polynomial time (${O(n^5)}$) on average for random graphs, which is consistent with our proof that the time complexity should never exceed ${O(n^7)}$.

\section{Concluding Remarks}

This paper has presented the polynomial time `uniqueness tree' algorithm, and used a combination of rigorous results and heuristic verification to propose that it may represent an effective approach to tackling the graph isomorphism problem.

The importance placed upon the graph isomorphism problem within computational complexity theory is primarily due to its present status as a prime candidate for an `${NP}$-intermediate' problem\cite{dawar} (a problem that is neither in the complexity class $P$, nor ${NP}$-complete), and many other proposed ${NP}$-intermediate problems are known to be reducible to GI.\cite{miller} Thus, if the graph isomorphism problem could be shown to be solvable in polynomial time, it would constitute evidence against the existence of the ${NP}$-intermediate set (and therefore, by Ladner's theorem, evidence that ${P = NP}$\cite{bodirsky}). As such, any indication of the possibility of a polynomial time solution for GI would have significant implications for the entire field.

Further research into the effectiveness of the uniqueness tree algorithm will centre around attempting to either prove or disprove the proposition that it is able to correctly discern the lack of an isomorphism in all relevant cases. Extensions of the methods described in this paper to non-simple and directed graphs will appear in a future work, along with possible generalisations of the uniqueness tree approach to related problems in graph theory (including graph automorphism, subgraph homeomorphism, etc.).

\section{Acknowledgements}

The author is grateful to Richard Bridges, for many helpful suggestions and clarifications during the early stages of this research.

\newpage
\appendix{Appendix: Source Code\\}

\begin{lstlisting}
import java.util.Random;
import java.io.*;

public class Main {

	private static Random rnd;
	private static FileWriter fileWriter;
	private static BufferedWriter bufferedWriter;

	private static int[] vertexCount;
	private static int[][] vertexDegree;
	private static int[][][] vertexAdjacency;

	private static int[] vertexMap;

	private static int[][] uniquenessTreeHeight;
	private static int[][][] uniquenessTreeWidth;
	private static int[][][][] uniquenessTreeVertices;
	private static int[][][][] uniquenessTreeVertexChildCount;

	private static int[][][][] uniquenessTreeVertexChildCountOccurrence;
	private static boolean[][] vertexMapped;

	public static void initialise (int n) {
		rnd = new Random ();

		vertexCount = new int[2];
		vertexDegree = new int[2][n];
		vertexAdjacency = new int[2][n][n];

		vertexMap = new int[100];

		uniquenessTreeHeight = new int[2][n];
		uniquenessTreeWidth = new int[2][n][n];
		uniquenessTreeVertices = new int[2][n][n][n * n];
		uniquenessTreeVertexChildCount = new int[2][n][n][n * n];

		uniquenessTreeVertexChildCountOccurrence = new int[2][n][n][n];
		vertexMapped = new boolean[2][n];
	}

	public static void generateRandomGraph (int size) {
		vertexCount[0] = size;

		for (int i = 0; i < size; i++) {
			for (int j = 0; j < size; j++) {
				if (i != j && rnd.nextDouble () <= 0.5) {
					boolean verticesAdjacent = false;

					for (int k = 0; k < vertexDegree[0][i]; k++) {
						if (vertexAdjacency[0][i][k] == j) {
							verticesAdjacent = true;
						}
					}

					if (!verticesAdjacent) {
						vertexAdjacency[0][i][vertexDegree[0][i]] = j;
						vertexAdjacency[0][j][vertexDegree[0][j]] = i;
						vertexDegree[0][i] += 1;
						vertexDegree[0][j] += 1;
					}
				}
			}
		}
	}

	public static void generateIsomorphicGraph () {
		vertexCount[1] = vertexCount[0];

		for (int i = 0; i < vertexCount[0]; i++) {
			vertexMap[i] = i;
		}

		for (int i = 0; i < vertexCount[0]; i++) {
			int swap = rnd.nextInt (vertexCount[0]);
			int temp = vertexMap[i];
			vertexMap[i] = vertexMap[swap];
			vertexMap[swap] = temp;
		}

		for (int i = 0; i < vertexCount[0]; i++) {
			vertexDegree[1][vertexMap[i]] = vertexDegree[0][i];
			for (int j = 0; j < vertexDegree[0][i]; j++) {
				vertexAdjacency[1][vertexMap[i]][j] = vertexMap[vertexAdjacency[0][i][j]];
			}
		}
	}

	public static void generateNonIsomorphicGraph () {
		generateIsomorphicGraph ();

		for (int i = 0; i < vertexCount[1]; i++) {
			if (vertexDegree[1][i] > 0) {
				vertexDegree[1][i] -= 1;
				vertexDegree[1][vertexAdjacency[1][i][vertexDegree[1][i]]] -= 1;
				
				for (int j = 0; j < vertexCount[1]; j++) {
					if (i != j && j != vertexAdjacency[1][i][vertexDegree[1][i]]) {
						boolean verticesAdjacent = false;
						
						for (int k = 0; k < vertexDegree[1][j]; k++) {
							if (vertexAdjacency[1][j][k] == i) {
								verticesAdjacent = true;
							}
						}
						
						if (!verticesAdjacent) {
							vertexAdjacency[1][i][vertexDegree[1][i]] = j;
							vertexAdjacency[1][j][vertexDegree[1][j]] = i;
							vertexDegree[1][i] += 1;
							vertexDegree[1][j] += 1;
							j = vertexCount[1];
							i = vertexCount[1];
						}
					}
				}
			}
		}
	}

	public static void generateUniquenessTrees () {
		for (int i = 0; i < 2; i++) {
			for (int j = 0 ; j < vertexCount[i]; j++) {

				uniquenessTreeVertices[i][j][0][0] = j;
				uniquenessTreeWidth[i][j][0] = 1;
				int uniqueVertexCount = 1;

				while (uniqueVertexCount > 0 && uniquenessTreeHeight[i][j] < (vertexCount[i] - 1)) {
					uniquenessTreeHeight[i][j] += 1;
					uniqueVertexCount = 0;

					for (int k = 0; k < uniquenessTreeWidth[i][j][uniquenessTreeHeight[i][j] - 1]; k++) {
						boolean vertexUnique = true;

						for (int l = 0; l < uniquenessTreeWidth[i][j][uniquenessTreeHeight[i][j] - 1]; l++) {
							if (k != l && uniquenessTreeVertices[i][j][uniquenessTreeHeight[i][j] - 1][k] == uniquenessTreeVertices[i][j][uniquenessTreeHeight[i][j] - 1][l]) {
								vertexUnique = false;
							}
						}

						if (vertexUnique) {
							uniqueVertexCount += 1;

							for (int l = 0; l < vertexDegree[i][uniquenessTreeVertices[i][j][uniquenessTreeHeight[i][j] - 1][k]]; l++) {
								uniquenessTreeVertices[i][j][uniquenessTreeHeight[i][j]][uniquenessTreeWidth[i][j][uniquenessTreeHeight[i][j]]] = 
										vertexAdjacency[i][uniquenessTreeVertices[i][j][uniquenessTreeHeight[i][j] - 1][k]][l];
								uniquenessTreeVertexChildCount[i][j][uniquenessTreeHeight[i][j] - 1][k] += 1;
								uniquenessTreeWidth[i][j][uniquenessTreeHeight[i][j]] += 1;
							}
						}
					}
				}
			}
		}
	}

	public static boolean computeIsomorphism () {
		for (int i = 0; i < 2; i++) {
			for (int j = 0; j < vertexCount[i]; j++) {
				for (int k = 0; k < uniquenessTreeHeight[i][j]; k++) {
					for (int l = 0; l < uniquenessTreeWidth[i][j][k]; l++) {
						uniquenessTreeVertexChildCountOccurrence[i][j][k][uniquenessTreeVertexChildCount[i][j][k][l]] += 1;
					}
				}
			}
		}

		for (int i = 0; i < vertexCount[0]; i++) {
			for (int j = 0; j < vertexCount[1]; j++) {

				if (!vertexMapped[0][i] && !vertexMapped[1][j]) {
					boolean verticesEquivalent =  true;

					if (uniquenessTreeHeight[0][i] != uniquenessTreeHeight[1][j]) {
						verticesEquivalent = false;
					}

					for (int k = 0; k < uniquenessTreeHeight[0][i]; k++) {
						if (uniquenessTreeWidth[0][i][k] != uniquenessTreeWidth[1][j][k]) {
							verticesEquivalent = false;
						}

						for (int l = 0; l < vertexCount[0]; l++) {
							if (uniquenessTreeVertexChildCountOccurrence[0][i][k][l] != uniquenessTreeVertexChildCountOccurrence[1][j][k][l]) {
								verticesEquivalent = false;
							}
						}
					}

					if (verticesEquivalent) {
						vertexMapped[0][i] = true;
						vertexMapped[1][j] = true;
					}
				}
			}
		}

		boolean graphsIsomorphic = true;

		for (int i = 0; i < vertexCount[0]; i++) {
			if (!vertexMapped[0][i]) {
				graphsIsomorphic = false;
			}
		}

		return graphsIsomorphic;
	}

	public static void main (String[] args) {
		int falseResults = 0;
		
		for (int i = 1; i < 100; i++) {
			System.out.println(i);
			long startTime = System.currentTimeMillis();

			for (int j = 1; j < 100; j++) {
				initialise (i);

				generateRandomGraph(i);
				generateIsomorphicGraph();
				//generateNonIsomorphicGraph();
				generateUniquenessTrees();
				if (!computeIsomorphism()) {
				//if (computeIsomorphism() && i > 9) {
					falseResults += 1;
					System.out.println ("False result! " + falseResults);
				}
			}

			try {
				fileWriter = new FileWriter ("output.csv", true);
				bufferedWriter = new BufferedWriter (fileWriter);
				bufferedWriter.write(Integer.toString(i) + ", " + Integer.toString((int)(System.currentTimeMillis() - startTime)));
				bufferedWriter.newLine();
				bufferedWriter.close();
			} catch (IOException ex) {
				System.out.println ("Error writing to output file.");
			}
		}
	}
}
\end{lstlisting}

\end{document}